\documentclass[11pt]{article}
\usepackage[english]{babel}
\usepackage{amsmath,amsthm}
\usepackage{amsfonts}
\usepackage[comma,numbers,square,sort&compress]{natbib}
\usepackage{graphicx,subfigure,amsmath,amssymb,mathrsfs,amsfonts,amstext,amsthm}
\usepackage{latexsym, amssymb}
\usepackage{graphicx}
\usepackage{subfigure}
\usepackage{pgf,tikz}
\usepackage{pstricks}
\newtheorem{thm}{Theorem}[section]

\newtheorem{lem}[thm]{Lemma}

\theoremstyle{definition}
\newtheorem{defn}[thm]{Definition}
\theoremstyle{remark}

\numberwithin{equation}{section}
\begin{document}

\title{\bfseries\textrm{Noise-based control of opinion dynamics*}
\footnotetext{*W. Su and X. Chen are with School of Automation and Electrical Engineering, University of Science and Technology Beijing \& Key Laboratory of Knowledge Automation for Industrial Processes, Ministry of Education, Beijing 100083, China, {\tt suwei@amss.ac.cn, cxz@ustb.edu.cn}. Yongguang Yu is with School of Science, Beijing Jiaotong University, Beijing 100044, China, {\tt ygyu@bjtu.edu.cn}. G. Chen is with National Center for Mathematics and Interdisciplinary Sciences \& Key Laboratory of Systems and
Control, Academy of Mathematics and Systems Science, Chinese Academy of Sciences, Beijing 100190,
China, {\tt chenge@amss.ac.cn}.}}

\author{Wei Su \and Xianzhong Chen \and Yongguang Yu \and Ge Chen }
%
\date{}%
\maketitle
\begin{abstract}
Designing feasible control strategies for opinion dynamics in complex social systems has never been an easy task. It requires a control protocol which 1) is not enforced on all individuals in the society, and 2) does not exclusively rely on specific opinion values shared by the social system. Thanks to the recent studies on noise-induced consensus in opinion dynamics, the noise-based intervention strategy has emerged as the only one meeting both of the above requirements, yet its underlying general theory is still lacking. In this paper, we perform rigorous theoretical analysis and simulations of a noise-based control strategy for opinion formation in which only a fraction of individuals is affected by randomly generated noise. We found that irrespective of the number of noise-driven individuals, including the case of only one single noise-affected individual, the system can attain a quasi-consensus in finite time, and the critical noise strength can be obtained. Our results highlight the efficiency of noise-driven mechanisms for the control of complex social dynamics.
\end{abstract}

\section{Introduction}
In the past decade, research on opinion dynamics and consensus problems in complex networked systems has drawn an increasing attention from a variety of fields, including mathematics, physics, social science, information theory, and various interdisciplinary areas \cite{Castellano2009,Proskurnikov2017,Olfati2004}.
In the study of complex social dynamics, developing control strategies for opinion consensus has been a central issue \cite{Friedkin2015, Liu2014,Albi2015,Albi2014,Wongkaev2015}. However, designing a practical and effective intervention strategy for a social opinion system that should attain global agreement has always been a grand challenge, largely due to two main obstacles. First, the traditional control theory has a practical shortcoming as it requires some kind of ``precise'' information about the system states, given that it is nearly impossible to acquire the accurate information about opinions of most individuals in a large social system. Secondly, due to the large system size, one is hardly capable of exerting control over every individual
in the society. Although the pinning control mechanism can usually handle such large complex networks by controlling only a fraction of agents, it still needs a precise information about the underlying system states \cite{Wang2002,Chen2007,Yu2009,Wang2014}. Considering these hindering issues, developing a new intervention strategy which can \emph{circumvent the reliance on system states} and instead \emph{act upon only a fraction of individuals} is indicated.

In recent years, a few noise-based control strategies for coordination enhancement and consensus generation in social networks have been advanced \cite{Shirado2017, Su2016arX, Chi2011}. For example, Shirado and Christakis \cite{Shirado2017} devised a network coordination experiment which took a local rule into consideration. After adding some noisy agents into the network, they observed a remarkable improvement of the coordination efficiency of the group.
Su {\it et al.} \cite{Su2016arX} designed a simple noise injection scheme to eliminate the disagreement in a divisive opinion system. By injecting random noise to only one agent of a divisive opinion system, generated by the local-rule based Hegselmann-Krause (HK) dynamics, it was shown that the cleavage of the system can be eliminated and the opinions can become synchronized.

These noise-based control strategies generally meet the requirements for a social system control; however, the underlying general theory for noise-induced opinion control is still lacking. For example, one limitation of the strategy of Su and colleagues \cite{Su2016arX} is that only one agent was allowed to receive noise, and moreover, the initial opinion state of the system had to be assumed {\it a priori} while the effective noise strength relied on the group size and was vanishingly small as the group size was increasing. The noise scheme developed in \cite{Su2016arX} was inspired by the previous seminal works on the noisy opinion dynamics \cite{Hadzibeganovic2008,Mas2010,Grauwin2012,Carro2013,Pineda2013,Lichtenegger2016,Hadzibeganovic2009,chazelle2017}, and was later theoretically verified for the HK opinion dynamics model \cite{Su2016}. However, the noise component in most of these previous models was typically added to all agents, which is far from realistic for a control of large-scale social systems. Critically, the mathematical approach developed in \cite{Su2016} was inadequate for modeling the scenario of noise application to only a fraction of agents.

Driven by these limitations of earlier studies, we herein aim to establish a general theoretical framework for the noise-based control strategy of large social networks in which no prior assumptions about the initial opinion states are required and where any percentage of individuals in the system can be influenced by noise. Currently, two major types of models for opinion dynamics can be distinguished. One is the class of topology-dependent models which follow the line of research initiated by DeGroot \cite{DeGroot1974}, and the other is the so-called bounded confidence model which originates from the works of Deffuant \emph{et al.} \cite{Deffuant2000} and Hegselmann and Krause \cite{Hegselmann2002}. The topology-dependent model is typically efficient in illustrating the opinion evolution of relatively small groups. For large social systems, however, the confidence-based model is much more convenient because it captures local-level self-organization processes that are the hallmark of large complex systems.

In this paper, we employ the confidence-based HK dynamics \cite{Hegselmann2002,Lorenz2007} to study the noise-based control theory of a large social opinion system. We rigorously prove that given any initial system states and any percentage of agents affected by noise, the system will almost surely attain a quasi-consensus in finite time. Crucially, our analysis yields the critical noise strength for quasi-consensus. More specifically, when only one agent is affected by noise, the critical noise strength is equal to the confidence threshold. If more than two agents are noise-driven, the critical noise strength is half the confidence threshold. These results provide a solid ground for the noise-based control strategy of social dynamics in large-scale systems.

The rest of the present paper is organized as follows: In section \ref{formulation}, we first present some necessary preliminaries. Our main findings are then presented and discussed in section \ref{Results}; in section \ref{Simulations}, we present our numerical simulation results to verify the main theoretical analyses and finally, some concluding remarks and future research directions are given in section \ref{Conclusions}.

\section{Definitions and model}\label{formulation}
Denote $\mathcal{V}=\{1,2,\ldots,n\}$ as the set of $n$ agents, $\mathcal{S}\subset\mathcal{V}$ the nonempty set of controlled agents. Let $x_i(t)\in[0,1], i\in\mathcal{V}, t\geq 0$ be the state of agent $i$ at time $t$ and $\xi_i(t), i\in\mathcal{S}, t>0$ be the noise control. The update rule for the noise-induced HK dynamics takes the form:
\begin{equation}\label{basicHKmodel}
  x_i(t+1)=\left\{
           \begin{array}{ll}
             1,  & \hbox{~$x_i^*(t)>1~$} \\
             x_i^*(t),  & \hbox{~$x_i^*(t)\in[0, 1]~$} \\
             0,  & \hbox{~$x_i^*(t)<0~$}
           \end{array},~\forall i\in\mathcal{V}, t\geq 0,
         \right.
\end{equation}
where
\begin{equation}\label{xit}
  x_i^*(t)=|\mathcal{N}(i, x(t))|^{-1}\sum\limits_{j\in \mathcal{N}(i, x(t))}x_j(t)+I_{i\in\mathcal{S}}\xi_i(t+1)
\end{equation}
and
\begin{equation}\label{neigh}
 \mathcal{N}(i, x(t))=\{j\in\mathcal{V}\; \big|\; |x_j(t)-x_i(t)|\leq \epsilon\}
\end{equation}
is the neighboring set of $i$ at $t$ with $\epsilon>0$ representing the confidence threshold of agents. Here, $I_{i\in\mathcal{S}}$ is the indicator function which takes the value 1 or 0 according to $i\in\mathcal{S}$ or not, and $|\cdot|$ stands for the cardinal number of a set or the absolute value of a real number.

To proceed further, we need to introduce some preliminary definitions as follows.
Let $\mathcal{G}_\mathcal{V}(t)=\{\mathcal{V},\mathcal{E}(t)\}$ be the graph of $\mathcal{V}$ at time $t$, and $(i,j)\in\mathcal{E}(t)$ if and only if $|x_i(t)-x_j(t)|\leq \epsilon$. A graph $\mathcal{G}_\mathcal{V}(t)$ is called a {\it complete graph} if and only if $(i,j)\in\mathcal{E}(t)$ for any $i\neq j$; and $\mathcal{G}_\mathcal{V}(t)$ is called a {\it connected graph} if and only if for any $i\neq j$, there is a sequence of edges $(i,i_1),(i_1,i_2),\ldots,(i_k,j)$ in $\mathcal{E}(t)$.

The definition of a \emph{quasi-consensus} of the noisy model (\ref{basicHKmodel})-(\ref{neigh}) is then as in \cite{Su2016}:
\begin{defn}\label{robconsen}
Denote
\begin{equation*}\label{opindist}
  d_{\mathcal{V}}(t)=\max\limits_{i, j\in \mathcal{V}}|x_i(t)-x_j(t)|~~\mbox{and}~~\overline{d}_{\mathcal{V}}=\limsup\limits_{t\rightarrow \infty}d_{\mathcal{V}}(t).
\end{equation*}
(i) if $\overline{d}_{\mathcal{V}} \leq \epsilon$, we say the system (\ref{basicHKmodel})-(\ref{neigh}) will reach a quasi-consensus.\\
(ii) if $P\{\overline{d}_{\mathcal{V}} \leq \epsilon\}=1$, we say almost surely (a.s.) the system (\ref{basicHKmodel})-(\ref{neigh}) will attain a quasi-consensus.\\
(iii) if  $P\{\overline{d}_{\mathcal{V}} \leq \epsilon\}=0$, we say a.s. the system (\ref{basicHKmodel})-(\ref{neigh}) cannot reach quasi-consensus.\\
(iv) let $T=\min\{t: d_{\mathcal{V}}(t')\leq \epsilon \mbox{ for all } t'\geq t\}$.
 If $P\{T<\infty\}=1$, we say a.s. the system (\ref{basicHKmodel})-(\ref{neigh}) attains a quasi-consensus in finite time.
\end{defn}

\section{Main Results}\label{Results}
For simplicity, we first present the result of a quasi-consensus for independent and identically distributed (i.i.d.) noises, and we then generalize these results with independent noises by a sufficient and a necessary condition.
\begin{thm}\label{consthm0}
Suppose the noises $\{\xi_i(t)\}_{i\in\mathcal{V},t\geq 1}$ are i.i.d. random variables with $E\xi_1(1)=0, 0<E\xi_1^2(1)<\infty$.
Let $x(0)\in [0,1]^n$ and $\epsilon\in(0,1]$ be arbitrarily given, then\\
 (i) if $P\{|\xi_1(1)|\leq \epsilon\}=1$ when $|\mathcal{S}|=1$ or $P\{|\xi_1(1)|\leq \epsilon/2\}=1$ when $|\mathcal{S}|>1$, then a.s. the system (\ref{basicHKmodel})-(\ref{neigh}) will attain a quasi-consensus in finite time; \\
 (ii) if $P\{\xi_1(1)>\epsilon\}>0$ and $P\{\xi_1(1)<-\epsilon\}>0$ when $|\mathcal{S}|=1$, or
    $P\{\xi_1(1)>\epsilon/2\}>0$ and $P\{\xi_1(1)<-\epsilon/2\}>0$ when $|\mathcal{S}|>1$,
then a.s. the system (\ref{basicHKmodel})-(\ref{neigh}) cannot reach a quasi-consensus.
\end{thm}
Conclusion (i) shows that if noise strength is no more than $\epsilon$ when $|\mathcal{S}|=1$ or $\epsilon/2$ when $|\mathcal{S}|>1$ a.s., the system will a.s. achieve a quasi-consensus in finite time; Conclusion (ii) states that when noise strength has a positive probability to exceed $\epsilon$ when $|\mathcal{S}|=1$ or $\epsilon/2$ when $|\mathcal{S}|>1$, the system will not reach quasi-consensus. This implies $\epsilon$ when $|\mathcal{S}|=1$ and $\epsilon/2$ when $|\mathcal{S}|>1$ are the critical noise strengths to induce a quasi-consensus. Conclusions (i) and (ii) can be directly derived from the following Theorems \ref{Suff_thm} and \ref{Nece_thm}, which present a sufficient and a necessary condition for independent noises, respectively.

\begin{thm}\label{Suff_thm}
Suppose $\{\xi_i(t), i\in\mathcal{V}, t\geq 1\}$ are independent and satisfy:
 i) $P\{|\xi_i(t)|\leq \delta\}=1$ with $\delta\in (0,\epsilon]$ when $|\mathcal{S}|=1$, or $\delta\in(0, \epsilon/2]$ when $|\mathcal{S}|>1$;
 ii) there exist constants $a\in(0,\delta),p\in(0,1)$ such that
$P\{\xi_i(t)\geq a\}\geq p, P\{0\leq \xi_i(t)\leq a\}\geq p$ and $P\{\xi_i(t)\leq -a\}\geq p, P\{-a\leq \xi_i(t)\leq 0\}\geq p$.
Then, for any initial state $x(0)\in [0,1]^n$ and  $\epsilon\in(0,1]$, the system (\ref{basicHKmodel})-(\ref{neigh}) will a.s. attain a quasi-consensus in finite time and $d_\mathcal{V}\leq \delta$ a.s. when $|\mathcal{S}|=1$, or $d_\mathcal{V}\leq 2\delta$ a.s. when $|\mathcal{S}|>1$.
\end{thm}

Before the proofs of Theorems \ref{consthm0} and \ref{Suff_thm} we need introduce some lemmas:

\begin{lem}\cite{Krause2000}\label{monosmlem}
Suppose $\{z_i, \, i=1, 2, \ldots\}$ is a nonnegative nondecreasing (nonincreasing) sequence. Then for any $s\geq 0$, the sequence
$\{g_s(k)=\frac{1}{k}\sum_{i=s+1}^{s+k}z_i$, $k\geq 1\}$ is monotonically nondecreasing (nonincreasing) for $k$.
\end{lem}
\begin{lem}\label{HKfrag}\cite{Blondel2009}
Suppose $\mathcal{S}=\emptyset$, then for every $1\leq i\leq n$ and $x_i(0)\in[0,1]$, there exist constant $T_0\geq 0, x_i^*\in[0,1]$ such that $x_i(t)=x_i^*$ for $t\geq T_0$, and either $x_i^*=x_j^*$ or $|x_i^*-x_j^*|> \epsilon$ holds for any $i, j$.
\end{lem}

In what follows, the ever appearing time symbols $t$ (or $T$, etc.) will all refer to the random variables $t(\omega)$ (or $T(\omega)$, etc.) on the probability space $(\Omega,\mathcal{F},P)$, and for simplicity, they will be still written as $t$ (or $T$).
%

In the rest process of the proofs of Theorems \ref{consthm0} and \ref{Suff_thm}, we only consider the case $|\mathcal{S}|>1$, and the proof for the case $|\mathcal{S}|=1$ can be obtained in a similar fashion.

\begin{lem}\label{robconspeci}
For the system (\ref{basicHKmodel})-(\ref{neigh}) with conditions of Theorem \ref{Suff_thm} i), if there exists a finite time $0\leq T<\infty$ such that $d_\mathcal{V}(T)\leq \epsilon$, then on $\{T<\infty\}$, we have
 $d_\mathcal{V}(t)\leq 2\delta$ for all $t> T$.
\end{lem}
\begin{proof}
Denote $\widetilde{x}_i(t)=|\mathcal{N}(i, x(t))|^{-1}\sum_{j\in \mathcal{N}(i, x(t))}x_j(t)$, $t\geq 0$, and this notation remains valid for the rest of the context. If $d_\mathcal{V}(T)\leq \epsilon$,
by (\ref{neigh}) we have
\begin{equation}\label{dtitera_0}
  \tilde{x}_i(T)=\frac{1}{n}\sum\limits_{j=1}^{n}x_j(T), \, \, i\in\mathcal{V}.
\end{equation}
Since $|\xi_i(t)|\leq \delta$ a.s., we obtain a.s.
\begin{equation}\label{dtitera}
\begin{split}
  d_\mathcal{V}(T+1)& =\max\limits_{1\leq i, j\leq n}|x_i(T+1)-x_j(T+1)| \\
    & \leq\max\limits_{1\leq i, j\leq n}(|I_{i\in\mathcal{S}}\xi_i(T+1)|+|I_{j\in\mathcal{S}}\xi_j(T+1)|)\\
    &\leq 2\delta\leq\epsilon.
\end{split}
\end{equation}
Repeating (\ref{dtitera_0}) and (\ref{dtitera}) yields the conclusion.
\end{proof}

\begin{lem}\label{clusterconsen}
For system (\ref{basicHKmodel})-(\ref{neigh}) with conditions of Theorem \ref{Suff_thm}, if at the initial moment there exist subsets $\mathcal{V}_1, \mathcal{V}_2\subset\mathcal{V}$ such that $\mathcal{V}_1\bigcup\mathcal{V}_2=\mathcal{V}$, $d_{\mathcal{V}_k}(0)\leq \epsilon, k=1,2$, and $ \mathcal{V}_1\bigcap\mathcal{V}_2=\emptyset$, $|x_i(0)-x_j(0)|>\epsilon$ for $i\in\mathcal{V}_1, j\in\mathcal{V}_2$, then there exist constants $L_0\geq 0$ and $p_0>0$ such that
$P\{d_\mathcal{V}(L_0)\leq \epsilon\}\geq p_0$ and $\min_ix_i(t)\geq\min_ix(0), \max_ix_i(t)\leq \max_ix_i(0)$ for $0\leq t\leq L_0$.
\end{lem}
\begin{proof}
This proof uses the idea that  ``transforming the analysis of a stochastic system into the design of control algorithms" first proposed
by \cite{Chen2017}.
At the initial moment, the systems forms 2 separate subgroups $\mathcal{V}_1, \mathcal{V}_2$ of which one is not neighboring with the other.
By (\ref{basicHKmodel}), $d_{\mathcal{V}_k}(1)\leq 2\delta\leq \epsilon, k=1,2$.
Before one subgroup enters the neighbor region of the other, for each $i\in \mathcal{V}$, we have
\begin{equation}\label{xitvk}
\begin{split}
  x_i(t+1)=\frac{1}{|\mathcal{V}_k|}\sum\limits_{j\in\mathcal{V}_k}x_j(t)+I_{i\in\mathcal{S}}\xi_i(t+1),
  \end{split}
\end{equation}
when $i\in\mathcal{V}_k, k=1,2$.
Suppose $\max_{i\in\mathcal{V}_1}x_i(0)<\min_{i\in\mathcal{V}_2}x_i(0)$, and consider the following noise protocol: for $t\geq 0$,
\begin{equation}\label{noiseproto0}
  \left\{
    \begin{array}{ll}
      \xi_i(t+1)\in[a,\delta], & \hbox{if}\quad i\in\mathcal{V}_1; \\
      \xi_i(t+1)\in[-\delta,-a], & \hbox{if}\quad i\in\mathcal{V}_2.
    \end{array}
  \right.
\end{equation}
Here $a\in (0,\delta)$ is a constant defined in Theorem \ref{Suff_thm}.
Under the protocol (\ref{noiseproto0}), it is easy to check that before one subgroup enters the neighbor region of the other, and by (\ref{xitvk}), we will have
\begin{equation}\label{extropindis}
\begin{split}
|\max_{i\in\mathcal{V}_2}x_i(t+1)-\min_{i\in\mathcal{V}_1}x_i(t+1)|<|\max_{i\in\mathcal{V}_2}x_i(t)-\min_{i\in\mathcal{V}_1}x_i(t)|-\frac{a}{n},
\end{split}
\end{equation}
suggesting that after each step, the maximum difference of opinion values decreases at least $\frac{a}{n}$.
This implies that there must exist a constant $\bar{L}_0\leq \frac{n(d_\mathcal{V}(0)-\epsilon)}{a}$ such that under the protocol (\ref{noiseproto0}),
\begin{equation}\label{firstmeet}
  \min_{i\in\mathcal{V}_2}x_i(\bar{L}_0)-\max_{i\in\mathcal{V}_1}x_i(\bar{L}_0)\leq \epsilon,\quad d_\mathcal{V}(\bar{L}_0)\leq 2\epsilon.
\end{equation}
From moment $\bar{L}_0+1$, design the following protocol: for $i\in\mathcal{V}, t\geq \bar{L}_0$,
\begin{equation}\label{noiseproto1}
  \left\{
    \begin{array}{ll}
      \xi_i(t+1)&\in[a,\delta],  \hbox{if}~\min\limits_{j\in\mathcal{V}}x_j(t)\leq\widetilde{x}_i(t)\leq \min\limits_{j\in\mathcal{V}}x_j(t)+\frac{d_{\mathcal{V}}(t)}{2}; \\
      \xi_i(t+1)&\in[-\delta,-a],  \hbox{if}~
      \min\limits_{j\in\mathcal{V}}x_j(t)+\frac{d_{\mathcal{V}}(t)}{2}<\widetilde{x}_i(t)\leq \max\limits_{j\in\mathcal{V}}x_j(t).
    \end{array}
  \right.
\end{equation}
Since $d_\mathcal{V}(\bar{L}_0)\leq \epsilon$ by (\ref{firstmeet}), we can check that for each $i\in\mathcal{S}$, under the protocol (\ref{noiseproto1}), either $x_i(t)\in[\min_{i\in\mathcal{V}}x_i(t)+a,\min_{i\in\mathcal{V}}x_i(t)+\epsilon]$ or $x_i(t)\in[\max_{i\in\mathcal{V}}x_i(t)-\epsilon,\max_{i\in\mathcal{V}}x_i(t)-a]$ or both hold, implying it is neighbor to agents with extreme opinion. By (\ref{basicHKmodel}), we know that (\ref{extropindis}) also holds. This implies there exists a constant $\bar{L}_1\leq \frac{n(d_\mathcal{V}(\bar{L}_0)-\epsilon)}{a}$ such that under the protocol (\ref{noiseproto1}),
\begin{equation}\label{finalmeet}
  \max_{i\in\mathcal{V}}x_i(\bar{L}_0+\bar{L}_1)-\min_{i\in\mathcal{V}}x_i(\bar{L}_0+\bar{L}_1)\leq \epsilon.
\end{equation}
By independence of $\xi_i(t), i\in\mathcal{V}, t\geq 1$, we know that the probability of the protocol (\ref{noiseproto0}) occurring $\bar{L}_0$ times and protocol (\ref{noiseproto1}) occurring $\bar{L}_1$ times is no less than $p^{|\mathcal{S}|(\bar{L}_0+\bar{L}_1)}>0$. Moreover, under the protocols (\ref{noiseproto0}) and (\ref{noiseproto1}), by Lemma \ref{monosmlem}, it holds $\min_ix_i(t)\geq\min_ix(0), \max_ix_i(t)\leq \max_ix_i(0)$ for $0\leq t\leq \bar{L}_0+\bar{L}_1$.
Let $L_0=\bar{L}_0+\bar{L}_1, p_0=p^{nL_0}$ and consider Lemma \ref{robconspeci}, then we obtain the conclusion.
\end{proof}

\begin{lem}\label{consposiprob}
Suppose the noise satisfies the conditions of Theorem \ref{Suff_thm} i), then for any $x(0)\in[0,1]^n, \epsilon\in(0,1]$, there exist a noise protocol and constants $L_n>0, 0<p_n<1$ such that $P\{d_\mathcal{V}(L_n)\leq \epsilon\}\geq p_n$. Furthermore, if $d_\mathcal{V}(0)>\epsilon$, it has under the noise protocol that $\min_ix_i(t)\geq\min_ix_i(0), \max_ix_i(t)\leq \max_ix_i(0)$ for $0\leq t\leq L_n$.
\end{lem}
\begin{proof}
If $d_\mathcal{V}(0)\leq \epsilon$, the conclusion holds directly from Lemma \ref{robconspeci}. Now we consider the case of $d_\mathcal{V}(0)> \epsilon$, and use the method of induction for the group size $n$. Note that there exist constants $a\in(0,\delta),p\in(0,1)$ such that
$P\{\xi_i(t)\geq a\}\geq p$ and $P\{\xi_i(t)\leq -a\}\geq p$. When $n=2$, consider the following noise protocol:
for $i\in\mathcal{S}$, if $x_i(0)\leq \frac{x_1(0)+x_2(0)}{2}$, take $\xi_i(1)\in[a,\delta]$; otherwise, take $\xi_i(1)\in[-\delta,-a]$. It can be easily seen that under the protocol, the opinion difference of the two agents will decrease at least $a$ for each step. If $d_\mathcal{V}(1)\leq \epsilon$, we obtain the conclusion by taking $L_2=1, p_2=p^2$ since the above protocol occurs with probability no less than $p^2$. Otherwise, let $\bar{L}=\frac{1-\epsilon}{a}$ and continue the above protocol $\bar{L}$ times, then we know that there must exist a constant $L_2\leq \bar{L}$ such that $d_\mathcal{V}(L_2)\leq \epsilon$. By independence of $\xi_i(t), i\in\mathcal{V}, t\geq 1$, we can take $p_2=p^{2\bar{L}}$, then the conclusion holds for $n=2$. Suppose the conclusion holds for $n=k\geq 2$, now we consider the case of $n=k+1$.

Let $\widetilde{x}_i(t)=|\mathcal{N}(i, x(t))|^{-1}\sum\limits_{j\in \mathcal{N}(i, x(t))}x_j(t)$, and consider the following noise protocol: For $i\in\mathcal{S}, t\geq 1$,
\begin{equation}\label{noiseproto00}
  \left\{
    \begin{array}{ll}
      \xi_i(t+1)&\in[a,\delta],  \hbox{if}~
       \min\limits_{j\in\mathcal{V}}x_j(t)\leq\widetilde{x}_i(t)\leq \min\limits_{j\in\mathcal{V}}x_j(t)+\frac{d_{\mathcal{V}}(t)}{2}; \\
      \xi_i(t+1)&\in[-\delta,-a], \hbox{if}~
      \min\limits_{j\in\mathcal{V}}x_j(t)+\frac{d_{\mathcal{V}}(t)}{2}<\widetilde{x}_i(t)\leq \max\limits_{j\in\mathcal{V}}x_j(t).
    \end{array}
  \right.
\end{equation}
By Lemma \ref{monosmlem}, it is easy to obtain that under the protocol (\ref{noiseproto00}), $\min_{j\in\mathcal{V}}x_j(t)$ is nondecreasing, and $\max_{j\in\mathcal{V}}x_j(t)$ is nonincreasing. Let $\bar{T}_0$ be the first moment when the graph of the system (\ref{basicHKmodel})-(\ref{neigh}) is not connected under the protocol (\ref{noiseproto00}).
We can prove that if the graph of (\ref{basicHKmodel})-(\ref{neigh}) remains connected under the protocol (\ref{noiseproto00}), the system will
attain a quasi-consensus in a constant period $\bar{L}_0=2n\frac{d_\mathcal{V}(0)-\epsilon}{\underline{d}}$ where $\underline{d}=\min\{a,d_0\}, 0<d_0<\frac{\epsilon}{2n^2+2n+1}$ (The proof of this fact will be given in Appendix as Lemma \ref{consenconnected}).
We then know that $P\{\bar{L}_0\leq \bar{T}_0<\infty\}=0$. Without loss of generality, suppose $\bar{T}_0=0$ a.s., i.e., at the initial moment, the system is not connected, and denote the two subgroups as $\mathcal{V}_1$ and $\mathcal{V}_2$; then no agent in group $\mathcal{V}_1$ is the neighbor of agents in $\mathcal{V}_2$, and vice versa.
Suppose $\max_{i\mathcal{V}_1}x_i(0)<\min_{i\in\mathcal{V}_2}x_i(0)$.

First consider the case when there is no agent controlled by noise in one of the two subgroups, say $\mathcal{V}_1$. By Lemma \ref{HKfrag}, there is a constant $\bar{L}_1$ such that $\mathcal{V}_1$ converges in $\bar{L}_1$.
By assumption, there is a constant $\bar{L}_2$ such that $\mathcal{V}_2$ reaches a quasi-consensus in $\bar{L}_2$ with a positive probability $\bar{p}_2$. If $\bar{L}_1\leq \bar{L}_2$, by Lemma \ref{clusterconsen}, there exist constants $L_0\geq 0, 0<p_0<1$ such that $P\{d_\mathcal{V}(\bar{L}_2+L_0)\leq\epsilon|d_{\mathcal{V}_2}(\bar{L}_2)\leq \epsilon\}\geq p_0$.
Let $L_{k+1}=\bar{L}_2+L_0, p_{k+1}=\bar{p}_2p_0$, then the conclusion holds for $n=k+1$. Otherwise, if $\bar{L}_1> \bar{L}_2$, consider the following noise protocol: for $i\in\mathcal{S}, \bar{L}_2<t\leq\bar{L}_1$, $\xi_i(t)\in[a,\delta]$. By Lemma \ref{monosmlem}, we know that during the period from $\bar{L}_2+1$ to $\bar{L}_1$, $\mathcal{V}_1$ and $\mathcal{V}_2$ will not enter the neighbor region of each other, and $\mathcal{V}_1$ reaches convergence in $\bar{L}_1$ while $\mathcal{V}_2$ keeps quasi-consensus. In other words, $\mathcal{V}_2$ reaches a quasi-consensus in $\bar{L}_1$ with a positive probability $\bar{p}_2p^{|\mathcal{S}|(\bar{L}_1-\bar{L}_2)}$. Recalling Lemma \ref{clusterconsen} and the former argument of the case $\bar{L}_1\leq \bar{L}_2$, we know the conclusion holds for $n=k+1$ when $\bar{L}_1> \bar{L}_2$. Hence the conclusion holds when there is only one subgroup controlled by noise.

Now we consider the case when both $\mathcal{V}_1$ and $\mathcal{V}_2$ are intervened by noise. By assumption, there exist constants $\bar{L}_1, \bar{L}_2>0$ and $0<\bar{p}_1,\bar{p}_2<1$, such that $\mathcal{V}_1$ reaches a quasi-consensus in $\bar{L}_1$ with probability $\bar{p}_1$ and $\mathcal{V}_2$ achieves a quasi-consensus in $\bar{L}_2$ with probability $\bar{p}_2$. If $\bar{L}_1= \bar{L}_2$, the conclusion holds by Lemma \ref{clusterconsen}. Suppose $\bar{L}_1\geq \bar{L}_2$ without loss of generality, and consider the following noise protocol: for $i\in\mathcal{S}\bigcap\mathcal{V}_2, \bar{L}_2<t\leq\bar{L}_1$, $\xi_i(t)\in[a,\delta]$. By Lemma \ref{monosmlem}, we know that during the period from $\bar{L}_2+1$ to $\bar{L}_1$, $\mathcal{V}_1$ and $\mathcal{V}_2$ will not enter the neighbor region of each other, and $\mathcal{V}_1$ reaches a quasi-consensus in $\bar{L}_1$ while $\mathcal{V}_2$ keeps quasi-consensus. Thus, they both reach quasi-consensus in $\bar{L}_1$ with a positive probability no less than $\bar{p}_1\bar{p}_2p^{|\mathcal{S}|(\bar{L}_1-\bar{L}_2)}$. By Lemma \ref{clusterconsen}, there exist constants $L_0\geq 0, 0<p_0<1$ such that $P\{d_\mathcal{V}(\bar{L}_1+L_0)\leq\epsilon|d_{\mathcal{V}_2}(\bar{L}_1)\leq \epsilon, d_{\mathcal{V}_1}(\bar{L}_1)\leq \epsilon\}\geq p_0$.
Let $L_{k+1}=\bar{L}_1+L_0, p_{k+1}=\bar{p}_1\bar{p}_2p^{|\mathcal{S}|(\bar{L}_1-\bar{L}_2)}p_0$, then the conclusion holds for $n=k+1$ when both of the subgroups are intervened by noise.

To sum up, the conclusion holds for $n=k+1$. This completes the proof.
\end{proof}

\noindent{\it Proof of Theorem \ref{Suff_thm}:} Define $U(L)=\{\omega:$ (\ref{basicHKmodel})-(\ref{neigh}) does not reach quasi-consensus in period $L\}$
and $U=\{\omega:$ (\ref{basicHKmodel})-(\ref{neigh}) does not reach quasi-consensus in finite time$\}$.
By Lemma \ref{consposiprob}, there exist constants $L_n>0, 0<p_n<1$ such that
$
  P\{U(L_n)\}\leq 1-p_n<1.
$
Since $x(0)$ is arbitrarily given in $[0,1]^n$, following the procedure of Lemma \ref{consposiprob}, it has
\begin{equation*}
  P\{U(mL_n)|U((m-1)L_n)\}\leq 1-p_n,\quad m>1.
\end{equation*}
Then
\begin{equation*}
\begin{split}
  P\{U\} =&P\Big\{\bigcap\limits_{m=1}^{\infty}U(mL_n)\Big\}=\lim\limits_{m\rightarrow \infty}P\{U(mL_n)\} \\
=&\lim\limits_{m\rightarrow \infty}\prod\limits_{k=1}^{m-1}P\{U((k+1)L_n)|U(kL_n)\}\cdot P\{U(L_n)\}\\
\leq& \lim\limits_{m\rightarrow \infty}(1-p_n)^m=0,
\end{split}
\end{equation*}
and hence,
\begin{equation*}
\begin{split}
   P\{\text{ (\ref{basicHKmodel})-(\ref{neigh}) can reach quasi-consensus in finite time}\}
   = 1-P\{U\}=1.
\end{split}
\end{equation*}
This completes the proof. \hfill $\Box$

Next, we will present the necessary part of the noise induced consensus, which shows that when the noise strength has a positive probability of exceeding $\epsilon/2$, the system a.s. cannot reach a quasi-consensus.
\begin{thm}\label{Nece_thm}
Let $x(0)\in [0,1]^n$, $\epsilon>0$ are arbitrarily given.
Assume the zero-mean random noises $\{\xi_i(t), i\in\mathcal{V},t\geq 1\}$ are i.i.d. with $E\xi_1^2(1)<\infty$ or independent with $\sup_{i,t}|\xi_i(t)|<\infty, a.s.$.
If there exists a lower bound $q>0$ such that
    $P\{\xi_i(t)>\epsilon/2\}\geq q$ and $P\{\xi_i(t)<-\epsilon/2\}\geq q$,
then a.s. the system (\ref{basicHKmodel})-(\ref{neigh}) cannot reach quasi-consensus.
\end{thm}
\begin{proof}
We only need to prove the independent case, while the i.i.d. case can be obtained similarly. For the independent case, we only need to prove that, for any constant $T_0\geq 0$, there exists $t\geq T_0$ a.s. such that $d_\mathcal{V}(t)>\epsilon$, i.e.
\begin{equation*}
  P\Big\{\bigcup\limits_{T_0=0}^{\infty}\{d_\mathcal{V}(t)\leq\epsilon, t\geq T_0\}\Big\}=0.
\end{equation*}
Given any $T_0\geq 0$, by independence of $\xi_i(t), i\in\mathcal{V}, t\geq 1$, it has
\begin{equation*}
\begin{split}
  P\{d_\mathcal{V}(T_0+1)>\epsilon\}\geq & P\{\min\limits_{x_i(T_0),i\in\mathcal{V}}\xi_i(T_0+1)<-\frac{\epsilon}{2},\max\limits_{x_i(T_0),i\in\mathcal{V}}\xi_i(T_0+1)>\frac{\epsilon}{2}\}\\
  \geq& q^2.
  \end{split}
\end{equation*}
Hence, $P\{d_\mathcal{V}(T_0+1)\leq \epsilon\}\leq 1-q^2<1$. Similarly,
\begin{equation*}
  P\Big\{d_\mathcal{V}(t)\leq\epsilon\Big|\bigcap\limits_{T_0\leq l<t}\{d_\mathcal{V}(l)\leq\epsilon\}\Big\}\leq 1-q^2.
\end{equation*}
Thus
\begin{equation*}
  \begin{split}
  P\{d_\mathcal{V}(t)\leq \epsilon, t\geq T_0\}=&P\Big\{\bigcap\limits_{t=T_0}^\infty\{d_\mathcal{V}(t)\leq\epsilon\}\Big\}=\lim\limits_{m\rightarrow\infty}P\Big\{\bigcap\limits_{t=T_0}^m\{d_\mathcal{V}(t)\leq\epsilon\Big\}\\
  =&\lim\limits_{m\rightarrow\infty}\prod\limits_{t=T_0}^mP\Big\{d_\mathcal{V}(t)\leq \epsilon\Big|\bigcap\limits_{l<t}\{d_\mathcal{V}(l)\leq\epsilon\}\Big\}\\
  \leq &\lim\limits_{m\rightarrow\infty}(1-q^2)^m=0.
  \end{split}
\end{equation*}
This completes the proof.
\end{proof}

\section{Numerical Simulations}\label{Simulations}
In this section, we present the outcomes of simulation experiments to verify our main theoretical results in this paper. First, we present a fragmentation of noise-free HK model. We take $n= 20, \epsilon=0.21$, and the initial opinion values are uniformly distributed on $[0,1]$. Fig. \ref{noise0fig} shows that the disagreement forms. We then generate random noise sources independently from a uniform distribution on $[-\delta,\delta] (\delta>0)$ and we randomly select half of agents to be intervened by noise. By Theorem \ref{Suff_thm}, when $\delta\leq 0.5\epsilon$, the system should achieve a quasi-consensus in finite time. We let $\delta=0.1\epsilon$, and we see in Fig. \ref{noise1fig} that the disagreement vanishes and the opinions become synchronized. Furthermore, we consider only one single agent to be controlled, while the other conditions remain unchanged. In Fig. \ref{per1noise1fig} we see that the system indeed attains a quasi-consensus in finite time.

\begin{figure}[htp]
  \centering
  \includegraphics[width=5in]{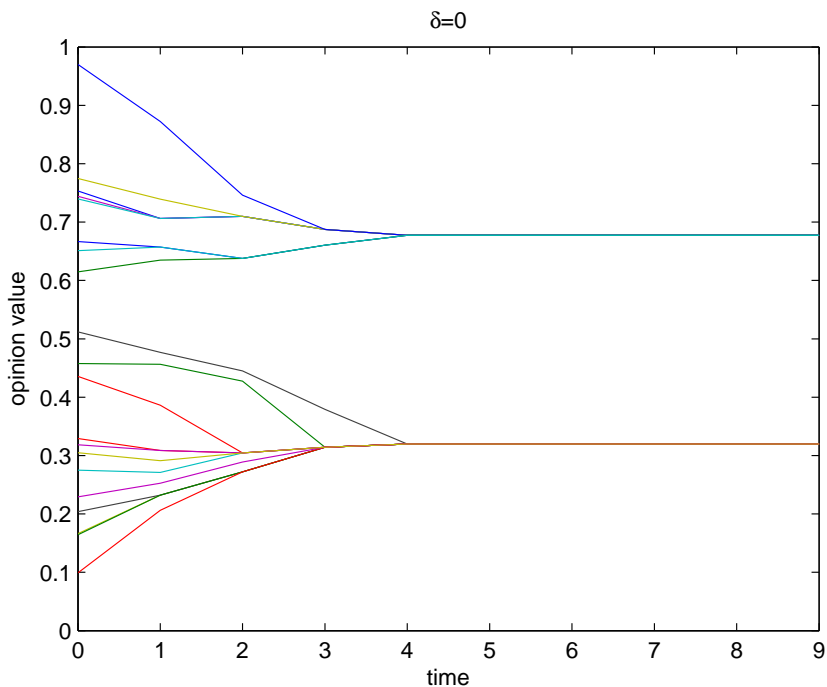}\\
  \caption{The evolution of opinions in a system with 20 agents without noise control. The initial opinion value $x(0)$ is generated by a uniform distribution on $[0,1]$, the confidence threshold $\epsilon=0.21$.  }\label{noise0fig}
\end{figure}
\begin{figure}[htp]
  \centering
  \includegraphics[width=5in]{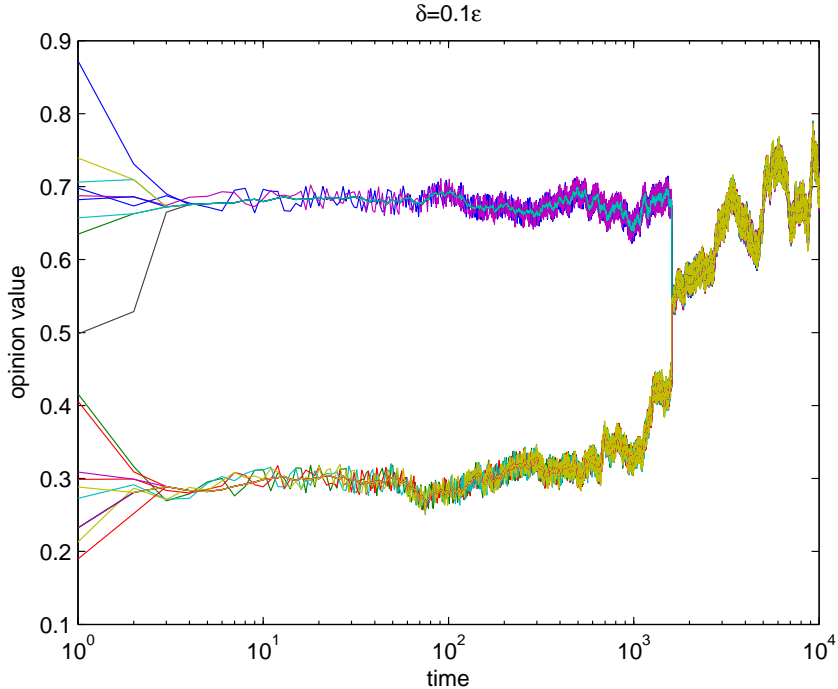}\\
  \caption{The evolution of opinions in a system with 20 agents and with a partial noise control. Here, the half of agents are randomly selected and subjected to the control, and the initial conditions are the same as those in Fig. \ref{noise0fig}; the noise strength $\delta=0.1\epsilon$.}\label{noise1fig}
\end{figure}
\begin{figure}[htp]
  \centering
  \includegraphics[width=5in]{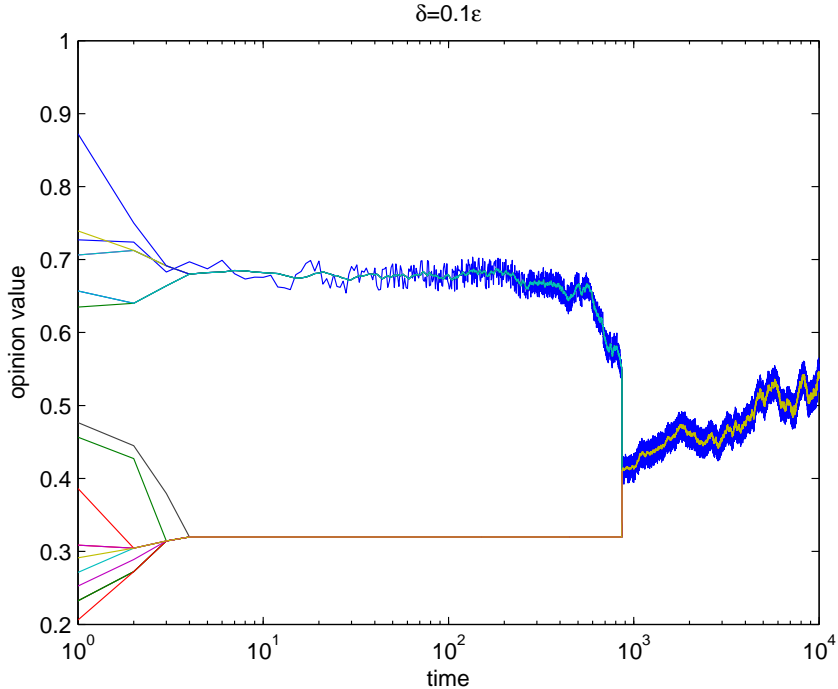}\\
  \caption{The evolution of opinions in a system with 20 agents and with a single-agent noise control. Only one agent is controlled, and the initial conditions are the same as those in Fig. \ref{noise0fig}, while the noise strength $\delta=0.1\epsilon$.  }\label{per1noise1fig}
\end{figure}
%

\section{Discussion  and Conclusions}\label{Conclusions}
Recently, HK opinion dynamics models with bounded confidence \cite{Hegselmann2002,Lorenz2007} and its variants \cite{Yyang2014,Fuetal2015,Wedin2015,Chazelle2017} have attracted considerable attention across fields. Nevertheless, feasible control strategies for HK systems are still lacking.

In this paper, we established a rigorous theoretical analysis for the noise-based opinion control strategy in large social networks where only a fraction of agents are noise-affected. The local rule based HK dynamics was taken as the underlying model of complex social systems. It was rigorously proved that, given any initial opinion configurations and irrespective of the amount of noise-affected agents, the noisy HK model will almost surely attain a quasi-consensus.

Indeed, the constructive role of noise has been revealed previously in a variety of physical, biological, and social systems (e.g. \cite{Su2016,Hadzibeganovic2015,Semenova2016}). However, its potential for the design of control strategies in dynamic social systems remained largely unexplored. While we have demonstrated that noise itself can serve as an efficient mechanism for control of opinion dynamics, its combined effects with other mechanisms are still unknown. For instance, recent studies of social dynamics \cite{Lichtenegger2016,HadStauffHan2018,Yamauchi2011} have evidenced nontrivial interactive effects of various combined mechanisms, including those with noise \cite{Lichtenegger2016,Hadzibeganovic2015}, that are not present when these mechanisms are considered in isolation. Future studies should therefore explore the effects of noise-related control strategies in social dynamics when they are employed in combination with other potent mechanisms.

For example, while previous theoretical analyses showed that the bounded confidence structure of HK dynamics can be synchronized by noise, it is still
unknown whether the noise-induced synchronization can occur with a topology-based mechanism. The theory of robust consensus in noisy multi-agent systems suggests that the topology-based mechanism does not have the property of noise-based synchronization. It is therefore necessary to examine the consequences of combining the topology-based mechanism with the bounded confidence structure.

In studies of noise-induced phenomena, it is typically assumed that noise source is a Gaussian distributed variable \cite{Wio2005}. However, it has been shown that noise-induced transitions can be shifted significantly when using non-Gaussian noise sources \cite{WioToral2004} characterized by nonextensive statistical properties \cite{Gellmann2004} that are often found in various biological and social processes \cite{Cannas2009,Kononovicius2014}. It thus remains a challenge for future research to also investigate noise-based mechanisms for the control of opinion dynamics when the noise source departs from the classical Gaussian behavior. Since non-Gaussian noise can significantly alter system's response by generating e.g. the reentrance effect (i.e. a shift from a disordered to an ordered state, and then back again to a disordered state)\cite{Wio2005,WioToral2004}, its applications could especially be interesting in systems in which only a temporary consensus is desired, after which the system can go back to its initial non-consensus state.

In sum, our present study demonstrated a feasibility of noise-induced control strategy in HK dynamic systems that can alter the system behavior after being applied to only a fraction of individuals and without relying on the complete knowledge of system's states. Our results thus represent the first step towards a more general theory of noise-induced control of social dynamics, and we hope they will inspire much further research.

\vspace{2.5ex}

\appendix

\begin{lem}\label{consenconnected}
Suppose the graph of the system (\ref{basicHKmodel})-(\ref{neigh}) remains connected under the protocol (\ref{noiseproto00}), then the system will
reach a quasi-consensus before a constant moment $\bar{L}_0=2n\frac{d_\mathcal{V}(0)-\epsilon}{\underline{d}}$ where $\underline{d}=\min\{a,d_0\}, 0<d_0<\frac{\epsilon}{2n^2+2n+1}$.
\end{lem}
\begin{proof}
If $d_\mathcal{V}(0)\leq \epsilon$, the conclusion holds by Lemma \ref{robconspeci}. Now we only consider the case of $d_\mathcal{V}(0)>\epsilon$. Denote $m(t)\in\mathcal{V}$ as the agent with the smallest opinion value at moment $t\geq 0$ and let $K(t)=\max\{1\leq j\leq n: x_j(t)-x_{m(t)}(t)\leq \epsilon\}$, then
\begin{equation}\label{x1tvalue}
  x_{m(t)}(t+1)=\frac{1}{K(t)}\sum\limits_{i=1}^{K(t)}x_i(t)+I_{m(t)\in\mathcal{S}}\xi_{m(t)}(t+1).
\end{equation}
If there exist $t_1,\ldots,t_n$ such that $m(t_i)\in\mathcal{S}, i=1,\ldots,n$, by (\ref{x1tvalue}) and Lemma \ref{monosmlem}, we know that $x_{m(t_n)}(t_n)$ increase at least $a$ from $x_{m(0)}(0)$ under the protocol (\ref{noiseproto00}). If an agent $i\in\mathcal{S}$ is a neighbor of agent $m(t)$ at moment $t$ and also at $t+1$, by (\ref{x1tvalue}) and Lemma \ref{monosmlem}, we know that at $t+1$, $x_{m(t+1)}(t+1)$ will increase at least $a/n$ under protocol (\ref{noiseproto00}). Now we consider the case when agents $m(t)$ and all its neighbors at $t$ are not in $\mathcal{S}$. Let $0<d_0<\frac{\epsilon}{2n^2+2n+1}$ be a constant. We will show that if there is a moment $t$ such that $x_{m(t+1)}(t+1)-x_{m(t)}t)\leq d_0$, then $x_{m(t+2)}(t+2)-x_{m(t+1)}(t+1)\geq d_0$. For convenience of denotation, we also number the agents at each moment $t$ with $x_1(t)\leq x_2(t)\leq \ldots\leq x_n(t)$. By (\ref{x1tvalue}), it has
\begin{equation*}
  x_{m(t+1)}(t+1)-x_{m(t)}(t)=\frac{1}{K(t)}\sum\limits_{i=1}^{K(t)}(x_i(t)-x_{m(t)}(t)),
\end{equation*}
then
\begin{equation}\label{distneigh}
  x_i(t)-x_{m(t)}(t)\leq K(t)d_0,\quad 1\leq i\leq K(t).
\end{equation}
Since the graph is connected, by the definition of $K(t)$ and (\ref{distneigh}), we know that $\epsilon-K(t)d_0<x_{K(t)+1}(t)-x_{K(t)}(t)\leq \epsilon$ where $K(t)+1$ is the agent with smallest opinion value larger than $x_{K(t)}(t)$. Hence by (\ref{basicHKmodel}) and Lemma \ref{monosmlem},
\begin{equation*}
  \begin{split}
  x_{K(t)}(t+1)-x_{K(t)}(t)
  =&\frac{1}{|\mathcal{N}(K(t), x(t))|}\sum\limits_{j\in \mathcal{N}(K(t), x(t))}(x_j(t)-x_{K(t)}(t)\\
  \geq & \frac{1}{K(t)+1}\Big(\sum\limits_{j=1}^{K(t)-1}(x_j(t)-x_{K(t)}(t))+x_{K(t)+1}(t)-x_{K(t)}(t)\Big)\\
  >& \frac{1}{K(t)+1}(\epsilon-K(t)d_0-(K(t)-1)K(t)d_0)\\
  =&\frac{1}{K(t)+1}\epsilon-\frac{K^2(t)}{K(t)+1}d_0\\
  >&(K(t)+1)d_0.
  \end{split}
\end{equation*}
hence $x_{K(t)}(t+1)-x_{m(t+1)}(t+1)\geq x_{K(t)}(t+1)-x_{m(t)}(t)-d_0\geq x_{K(t)}(t+1)-x_{K(t)}(t)-d_0>K(t)d_0$.
This implies that at $t+1$ the opinion value of agent $K(t)$ is more than $K(t)d_0$ apart from agent $m(t+1)$. If agent $K(t)$ is still a neighbor of agent $m(t+1)$ at $t+1$, by (\ref{distneigh}), we have $x_{m(t+2)}(t+2)-x_{m(t+1)}(t+1)\geq d_0$. Otherwise, repeating the above procedure no more than $n$ times. Since the graph is always connected, there exists a moment $t\leq t_h\leq t+n$ such that $K(t_h)$ is the neighbor of $m(t_h+1)$, implying $x_{m(t_h+1)}(t_h+1)-x_{m(t)}(t)\geq x_{m(t_h+1)}(t_h+1)-x_{m(t_h)}(t_h)\geq d_0$. Hence, after each moment when the opinion value of agent $m(t)$ increases less than $d_0$, it will increase at least by $d_0$ during the next $n$ times. Let $\underline{d}=\min\{a,d_0\}, \bar{L}_0=2n\frac{d_\mathcal{V}(0)-\epsilon}{\underline{d}}$, we know with the above analysis that $d_\mathcal{V}(\bar{L}_0)\leq \epsilon$, then by Lemma \ref{robconspeci}, the system will reach a quasi-consensus.
\end{proof}

\end{document}